\pdfoutput=1

\documentclass[twocolumn,nofootinbib,preprintnumbers]{revtex4-1}

\usepackage[vmargin = 0.85in, hmargin = 0.6in]{geometry}

\usepackage{amsmath,amsthm,amssymb,amsfonts}

\usepackage{array}

\usepackage{multirow}

\usepackage{mathtools}
\mathtoolsset{centercolon}

\usepackage{etoolbox}

\newcommand{\refx}[1]{\ref{#1}}

\usepackage{tikz}
\usetikzlibrary{calc,positioning,arrows,shapes.symbols,shapes.geometric}

\usepackage[colorlinks = true]{hyperref}

\hypersetup{
  pdftitle = {Bound entangled states with private key and their classical counterpart},
  pdfauthor = {Maris Ozols, Graeme Smith, John A.~Smolin}
}

\usepackage{color}
\definecolor{darkred}  {rgb}{0.5,0,0}
\definecolor{darkblue} {rgb}{0,0,0.5}
\definecolor{darkgreen}{rgb}{0,0.5,0}
\hypersetup{
  urlcolor   = blue,         
  linkcolor  = darkblue,     
  citecolor  = darkgreen,    
  filecolor  = darkred       
}

\def\step{18pt}
\def\dangle{15}

\newcommand{\grid}[2]{
  \foreach \i in {0,...,#2}{
    \draw[gridline] (\i*\step-\step/2,\step/2) -- +(0,-#1*\step);
  };
  \foreach \i in {0,...,#1}{
    \draw[gridline] (-\step/2,-\i*\step+\step/2) -- +(#2*\step,0);
  };
  \pgfmathparse{#2-1}
  \foreach \i in {0,...,\pgfmathresult}{
    \node (a\i) at (\i*\step,\step) {$\i$};
  };
  \pgfmathparse{#1-1}
  \foreach \i in {0,...,\pgfmathresult}{
    \node (b\i) at (-\step,-\i*\step) {$\i$};
  };
  \node at (#2*\step/2-\step/2,1.6*\step) {Bob};
  \node[rotate = 90] at (-1.6*\step,-#1*\step/2+\step/2) {Alice};
}

\newcommand{\nodes}[2]{
  \pgfmathparse{#1-1}
  \foreach \i in {0,...,\pgfmathresult} {
    \pgfmathparse{#2-1}
    \foreach \j in {0,...,\pgfmathresult}
      \node[circ] (\i\j) at (\j*\step,-\i*\step) {};
  }
}

\newcommand{\drawAB}[3]{
\begin{tikzpicture}[
  gridline/.style = {black},
  circ/.style = {circle, draw = black, fill = black,
                 inner sep = 0mm, minimum size = 0.13*\step},
  state/.style = {black, thick}]
  \grid{#1}{#2}
  \nodes{#1}{#2}
  #3
\end{tikzpicture}
}

\newcommand{\emptynode}[1]{\node[circ, white, text = black] at (#1) {\yuzz}}

\newtheorem{theorem}{Theorem}
\newtheorem{corollary}[theorem]{Corollary}
\newtheorem{lemma}{Lemma}
\theoremstyle{definition}
\newtheorem*{example}{Example}

\newcommand{\fnote}[1]{\footnote{#1}}

\DeclarePairedDelimiter{\set}{\lbrace}{\rbrace}
\DeclarePairedDelimiter{\abs}{\lvert}{\rvert}

\newcommand{\pt}{^{\mathsf{\Gamma}}} 

\DeclareMathOperator{\tr}{Tr}

\newcommand{\mx}[1]{\begin{pmatrix}#1\end{pmatrix}}

\newcommand{\ket}[1]{|#1\rangle}

\newcommand{\proj}[1]{|#1\rangle\langle#1|}
\newcommand{\ketbra}[2]{|#1\rangle\langle#2|}

\input{"fig-yuzz.tex"}
\newcommand{\yuzz}{
\begin{tikzpicture}[scale=0.00535]
  \yuzzcode
\end{tikzpicture}}

\newcommand{\ECPA}{EC+PA}

\csdef{3x3}{0.0057852}
\csdef{4x4}{0.0293914}
\csdef{4x5}{0.0480494}
\csdef{5x6}{0.0378462}
\csdef{6x5}{0.0354342}
\csdef{4H4}{0.0213399}

\begin{document}


\title{Bound entangled states with private key and their classical counterpart}

\author{Maris Ozols}
\author{Graeme Smith}
\author{John A. Smolin}
\affiliation{IBM TJ Watson Research Center, 1101 Kitchawan Road, Yorktown Heights, NY 10598}

\begin{abstract}
Entanglement is a fundamental resource for quantum information processing. In its pure form, it allows quantum teleportation
and sharing classical secrets. Realistic quantum states are noisy and their usefulness is only partially understood. Bound-entangled states are central to this question---they have no distillable entanglement, yet sometimes still have a private classical key. We present a construction of bound-entangled states with private key based on classical probability distributions.  From this emerge states possessing a new classical analogue of bound entanglement, distinct from the long-sought bound information. We also find states of smaller dimensions and higher key rates than previously known. Our construction has implications for classical cryptography: we show that existing protocols are insufficient for extracting  private key from our distributions due to their ``bound-entangled'' nature.  We propose a simple extension of existing protocols that can extract key from them.
\end{abstract}

\maketitle

\section{Introduction} \label{sect:Introduction}

A fundamental goal of cryptography is to establish secure communication between two parties, Alice and Bob, in the presence of an eavesdropper Eve. This can be achieved by allowing Alice and Bob to encrypt their communication using a key obtained from some initially shared resource---a joint probability distribution~\cite{Shannon49, Wyner75, CK78} or quantum state~\cite{E91, LoChau99}. However, this resource may not be useful in its original form---the shared key may not be perfectly random, private, or identical for both parties. Thus, classical key distillation---the process of generating perfectly random, private and identical key from a given tripartite probability distribution $P_{ABE}$ shared among the three parties---and the similar quantum task of entanglement distillation from a tripartite quantum state $\ket{\psi}_{ABE}$, are problems of fundamental importance~\cite{Wyner75, CK78, BBR88, AhlswedeCsiszar, Maurer, BBCM95, BDSW96}.

Private key is weaker than entanglement---it can be obtained by measuring Einstein-Podolsky-Rosen (EPR) pairs~\cite{E91}. Thus, one can distill private key from a quantum state by first distilling EPR pairs. This strategy is not optimal in general due to the existence of \emph{private bound entanglement}---entangled states from which EPR pairs cannot be distilled, but nevertheless a private key can be obtained~\cite{BoundKey, SmallBoundKey}.

Private bound-entangled states demonstrate a striking distinction between two forms of correlation: private classical key and shared entanglement. The best rate at which Alice and Bob can generate private key from $\ket{\psi}_{ABE}$ when only Eve has access to the $E$ part of the initial state and all public messages sent by Alice and Bob is called the \emph{private key rate}, denoted by $K(\psi_{ABE})$. Now, in addition to the above, assume that Eve also has access to all ancillary trash systems that Alice and Bob have introduced during the protocol. \emph{I.e.}, all information produced during the protocol---other than the final key---becomes available to Eve once the protocol is over. In such a case, distilling key becomes much harder. In the language of \cite{BoundKey, SmallBoundKey}, the final ``key'' system cannot be protected by any ``shield'' systems kept by Alice and Bob.  In fact, as the following simple observations imply, Alice and Bob have no choice but resort to distilling a much stronger resource---entanglement. One can check that:
\def\leftmargini{3ex}
\begin{itemize}
  \item if Alice and Bob can distill entanglement, they maintain privacy from Eve even if she has access to all ancillary trash systems produced during the protocol;
  \item conversely, the only way of obtaining a resource that guarantees privacy between Alice and Bob when all trash systems are available to Eve is to distill entanglement\fnote{Trash systems include all purifying systems too, hence a shared classical key is not private from Eve in this setting.}.
\end{itemize}
Thus, the best rate of producing a private key in the more restricted scenario when all ancillary trash systems are available to Eve, is the same as \emph{entanglement distillation rate} $D(\psi_{ABE})$---the best rate at which Alice and Bob can distill EPR pairs from $\ket{\psi}_{ABE}$ via local operations and classical communication (LOCC).  Private bound-entangled states have $D(\psi_{ABE}) = 0$ and $K(\psi_{ABE}) > 0$.

\newcommand{\cd}[1]{\multicolumn{1}{|c|}{#1}}
\newcommand{\cb}[1]{\multicolumn{1}{c|}{#1}}
\def\rsep{-0.9ex}
\begin{table}[htb]
\begin{tabular}[c]{|p{2.05cm}|p{2.05cm}|p{4.00cm}|}
\hline &&\\[\rsep]
  &
  \cb{\textit{Quantum}} &
  \cb{\textit{Classical}} \\[4pt]
\hline &&\\[\rsep]
  \cd{\multirow{3}{*}{\textit{States}}}
  & \cb{$\ket{\psi}_{ABE}$} & \cb{$P_{ABE}$} \\[5pt]
  & \cb{unambiguous}        & \cb{unambiguous probability} \\
  & \cb{quantum state}      & \cb{distribution} \\[4pt]
\hline &&\\[\rsep]
  \cd{\multirow{2}{*}{\textit{Entanglement}}}
  & \cb{$D(\psi_{ABE})$}     & \cb{$K_{\rm PD}(P_{ABE})$} \\[5pt]
  \cd{\textit{distillation}}
  & \cb{EPR pairs}          & \cb{private key by public} \\
  \cd{(\textit{public trash})}
  & \cb{by LOCC}            & \cb{discussion} \\[4pt]
\hline &&\\[\rsep]
  \cd{\multirow{2}{*}{\it{Private key}}}
  & \cb{$K(\psi_{ABE})$}     & \cb{$K(P_{ABE})$} \\[5pt]
  \cd{\textit{distillation}}
  & \cb{private key}         & \cb{private key by public discus-} \\
  \cd{(\textit{private trash})}
  & \cb{by LOCC}            & \cb{sion and noisy processing} \\[4pt]
\hline
\end{tabular}
\caption{\label{tab:Dictionary}Quantum-classical dictionary for states and distillation rates.  A tripartite probability distribution $P_{ABE}$ is unambiguous if it satisfies Eqs.~(\ref{eq:OliveA}--\ref{eq:OliveE}).  The associated quantum state $\ket{\psi_{ABE}}$ is given by Eq.~\eqref{eq:psiABE}.}
\end{table}

We will show that a similar distinction exists also in the classical world.  In the classical case, a private key must be distilled from a shared probability distribution $P_{ABE}$ by public discussion between Alice and Bob.  At each step of the protocol, either Alice or Bob generates a public message from her/his random variables, followed by a stochastic map that modifies the variables.  In general, such map might not be reversible and thus partially destroy the information (we call such maps \emph{noisy processing}).  We denote by $K(P_{ABE})$ the best private key rate obtainable by such protocols (\emph{i.e.}, protocols that involve public discussion and noisy processing).  In an alternative scenario, Alice and Bob can only create new random variables but cannot modify or destroy the existing ones\fnote{This does not impose any restrictions by itself.  The crucial difference is that all auxiliary variables must be surrendered to Eve at the end of the protocol.}. Furthermore, all variables (except the ones that contain the key) become available to Eve at the end of the protocol.  We denote the best key rate of such protocols by $K_{\rm PD}(P_{ABE})$, where PD stands for \emph{public discussion} (the protocol involves only public discussion and no noisy processing).  Because in the quantum setting it is \emph{distillable entanglement} that is resistant to giving trash systems to Eve, its natural classical analogue is $K_{\rm PD}$.  The
quantum quantity corresponding to the private key achieved by including noisy processing $K(P_{ABE})$ is simply the private key obtainable by LOCC, $K(\psi_{ABE})$.  Table~\ref{tab:Dictionary} summarizes the quantities of interest.

Previous studies pursuing a classical analogue of bound entanglement~\cite{GRW02, RW03, ACM04, PA12} looked for distributions with $K(P_{ABE}) = 0$.  A particular distribution, obtained by measuring a bound-entangled quantum state, was considered in~\cite{GRW02}.  It was hoped that because the quantum state was bound, no key would be distillable from  the classical distribution.  This hope was tempered by the discovery of private bound-entangled states~\cite{BoundKey, SmallBoundKey}, whose existence demonstrates a clear distinction between secrecy and bound entanglement.  Our work establishes a similar distinction classically by giving distributions with $K_{\rm PD}(P_{ABE}) = 0$ that cannot be created by public discussion, in direct analogy with quantum bound-entangled states.  We specifically do not solve the longstanding question of whether or not there is \emph{bound information}~\cite{GRW02, RW03, ACM04, PA12}, which corresponds to $K(P_{ABE}) = 0$ according to our notation (see Table~\ref{tab:Dictionary}) and which we would prefer to call \emph{bound private key}.  It is interesting to note that in the tripartite case, an affirmative answer has been demonstrated classically~\cite{ACM04}.

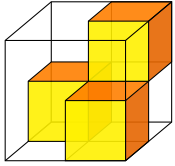
\begin{figure}


\def\aa{40} 
\def\dd{0.5} 

\newcommand{\cube}[6]{
  \path (#1,#2) ++(\aa:#3*\dd) coordinate (v00#6);
  \path (v00#6) +(#4,0 ) coordinate (v01#6);
  \path (v00#6) +(0 ,#4) coordinate (v10#6);
  \path (v00#6) +(#4,#4) coordinate (v11#6);
  \path (v00#6) +(\aa:#4*\dd) coordinate (w00#6);
  \path (v01#6) +(\aa:#4*\dd) coordinate (w01#6);
  \path (v10#6) +(\aa:#4*\dd) coordinate (w10#6);
  \path (v11#6) +(\aa:#4*\dd) coordinate (w11#6);
  \draw (w10#6) -- (w00#6) -- (w01#6);
  \draw (v00#6) -- (w00#6);
  {#5}
  \draw (v00#6) -- (v10#6) -- (w10#6) -- (w11#6) -- (w01#6) -- (v01#6) -- cycle;
  \draw (v10#6) -- (v11#6) -- (v01#6);
  \draw (v11#6) -- (w11#6);
}

\newcommand{\bigcube}[1]{
  \begin{scope}[thin, scale = 0.8]
    \cube{0}{0}{0}{2}{#1}{b}
  \end{scope}
}

\definecolor{darkorange}{rgb}{0.9, 0.4, 0.0}

\newcommand{\smallcube}[3]{
  \begin{scope}[thin]
  \cube{#1}{#2}{#3}{1}{
    \draw[fill = yellow,     draw = none, opacity = 0.9] (v00s) -- (v01s) -- (v11s) -- (v10s);
    \draw[fill = darkorange, draw = none, opacity = 0.9] (v01s) -- (w01s) -- (w11s) -- (v11s) -- (v01s);
    \draw[fill = orange,     draw = none, opacity = 0.9] (v10s) -- (w10s) -- (w11s) -- (v11s) -- (v10s);
  }{s}
  \end{scope}
}


\begin{tikzpicture}
  \bigcube{
    \smallcube{0}{0}{1}
    \smallcube{1}{0}{0}
    \smallcube{1}{1}{1}
  }
\end{tikzpicture}

\caption{\label{fig:Cubes}A three-dimensional representation of an unambiguous probability distribution $P_{ABE}$.  Each axis corresponds to one of the three parties and each cube represents a triple $(a,b,e)$ such that $p(a,b,e) \neq 0$.  Intuitively, Eqs.~(\ref{eq:OliveA}--\ref{eq:OliveE}) say that the small cubes do not overlap if this block is compressed along any of the three axis.}
\end{figure}

\section{Construction} \label{sect:Construction}

\begin{figure}[!ht]
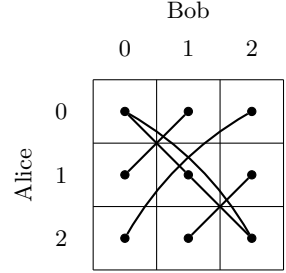


  \def\step{24pt}
\drawAB{3}{3}{
  \draw[state] (00) -- (11) -- (22);
  \draw[state] (01) -- (10);
  \draw[state] (12) -- (21);
  \draw[state] (02) .. controls +(225-\dangle:\step) and +(45+\dangle:\step).. (20);
  \draw[state] (00) .. controls +(-45+\dangle:\step) and +(135-\dangle:\step).. (22);
}

\caption{\label{fig:Olives}A graphical representation of an unambiguous distribution $P_{ABE}$ with $d_A = d_B = 3$ and $d_E = 4$. For each nonzero entry of $p(a,b,e)$ we put a dot at coordinates $(a,b)$, and connect dots corresponding to the same symbol for Eve. Note that each cell contains at most one dot due to Eq~\eqref{eq:OliveE}, and the resulting graph is a union of disjoint cliques (complete graphs) where each clique represents a different symbol for Eve. The above example has four connected components, hence $d_E = 4$. Furthermore, no two vertices from a clique share the same column or row due to Eqs.~\eqref{eq:OliveA} and~\eqref{eq:OliveB}. For PT-invariance (see Appendix~\refx{apx:PT-invariance}), the diagram in addition must also be a union of \emph{crosses}, \emph{i.e.}, pairs of edges $(a,b)\!-\!(a',b')$ and $(a,b')\!-\!(a',b)$ for some $a \neq a'$ and $b \neq b'$. The above diagram consists of three crosses: two small and one large.}
\end{figure}

Our results are based on tripartite probability distributions $P_{ABE}$ whose probabilities $p(a,b,e)$ have a special combinatorial structure\fnote{Distributions with similar properties have been considered before: tripartite distributions that satisfy only Eq.~(\ref{eq:OliveE}) appeared in~\cite{CEHHOR07}; bipartite distributions with similar properties (called bi-disjoint distributions) appeared in~\cite{NegativeClassical}.} (see Fig.~\ref{fig:Cubes}):
\begin{align}
  \forall b,e\quad \abs{\set{a : p(a,b,e) \neq 0}} & \leq 1 \label{eq:OliveA}, \\
  \forall a,e\quad \abs{\set{b : p(a,b,e) \neq 0}} & \leq 1 \label{eq:OliveB}, \\
  \forall a,b\quad \abs{\set{e : p(a,b,e) \neq 0}} & \leq 1 \label{eq:OliveE},
\end{align}
where $\abs{S}$ denotes the size of set $S$.
We call such distributions \emph{unambiguous}, since any two parties can uniquely determine the third party's variable. Such distributions have a convenient graphical representation (see Fig.~\ref{fig:Olives}), which together with $P_{AB}$ determines the full distribution $P_{ABE}$ (up to permutations on $E$).

We identify $P_{ABE}$ with a tripartite pure state
\begin{equation}
  \ket{\psi}_{ABE}
  := \sqrt{P_{ABE}}
  := \sum_{a,b,e} \sqrt{p(a,b,e)}
     \ket{a}_A \ket{b}_B \ket{e}_E,
  \label{eq:psiABE}
\end{equation}
where $\ket{a}_A$, $\ket{b}_B$, $\ket{e}_E$ are standard basis vectors for systems $A,B,E$, and with a bipartite mixed state
\begin{equation}
  \rho_{AB} := \tr_E \proj{\psi}_{ABE}
  \label{eq:rhoAB}
\end{equation}
on Alice and Bob whose purification is held by Eve. Such states have a special structure, since all eigenvectors of $\rho_{AB}$ have the same Schmidt basis.

The \emph{partial transpose}\fnote{Normally one has to specify the system on which the partial transpose is performed. However, all our density matrices are real (and thus symmetric), so the partial transpose on Alice's side is equivalent to partial transpose on Bob's side.} (PT) of $\rho_{AB}$ is defined on the standard basis as
\begin{equation}
  \bigl( \ketbra{a}{a'}_A \otimes \ketbra{b}{b'}_B \bigr)\pt
  := \ketbra{a}{a'}_A \otimes \ketbra{b'}{b}_B
  \label{eq:PT}
\end{equation}
and extended by linearity. If $\rho_{AB}$ is \emph{PT-invariant} ($\rho_{AB}\pt = \rho_{AB}$) then it has positive partial transpose and thus no distillable entanglement~\cite{Bound} (unambiguous distributions $P_{ABE}$ that yield PT-invariant states $\rho_{AB}$ are characterized in Appendix~\refx{apx:PT-invariance}).

\section{Results}

Using the properties of unambiguous distributions and Eq.~(\ref{eq:psiABE}), which promotes any classical distribution to a quantum state, we establish a strong analogy between classical and quantum distillation problems in Table~\ref{tab:Dictionary}.  We show that a classical protocol with an unambiguous initial distribution can be ``lifted'' to a quantum protocol with an unambiguous initial state, without decreasing the associated distillation rate.

\begin{theorem}\label{thm:Main}
Let $P_{ABE}$ be an unambiguous probability distribution and $\ket{\psi}_{ABE}$ be the associated quantum state. The distillable entanglement of $\ket{\psi}_{ABE}$ is at least as big as the distillable key by public discussion of $P_{ABE}$,
\begin{equation}
  D(\psi_{ABE}) \geq K_{\rm PD}(P_{ABE}).
\end{equation}
The distillable key of $\ket{\psi}_{ABE}$ is at least as big as the distillable key by public discussion and noisy processing of $P_{ABE}$\fnote{This was first observed in~\cite{CEHHOR07} (see supplementary material for more details).},
\begin{equation}
  K(\psi_{ABE}) \geq K(P_{ABE}).
\end{equation}
\end{theorem}

\begin{proof}[Proof sketch.]
Let $\sqrt{Q_{ABE}}$ be the quantum state associated to distribution $Q_{ABE}$ at some step of the classical protocol, and let $\sqrt{M}$ denote the entry-wise square root of the stochastic map $M$ that is applied.  Without loss of generality, $M$ introduces a new random variable.  Hence, if $Q_{ABE}$ is unambiguous then so is  $M \cdot Q_{ABE}$.  By induction, the distribution remains unambiguous throughout the protocol.  Furthermore, at every step $\sqrt{M} \cdot \sqrt{Q_{ABE}} = \sqrt{M \cdot Q_{ABE}}$, which allows lifting the classical protocol to a quantum one.  The quantum protocol achieves the same rate due to properties of unambiguous states (see Appendix~\refx{apx:Correspondence} for complete proof).
\end{proof}

Recall that private bound-entangled states have $D(\psi_{ABE}) = 0$ and $K(\psi_{ABE}) > 0$, implying that entanglement and private key are distinct resources in the quantum world. We show that $K_{\rm PD}$ and $K$ also correspond to distinct resources classically.

\begin{theorem}\label{thm:cPBE}
Unambiguous probability distributions $P_{ABE}$ with $K_{\rm PD}(P_{ABE}) = 0$ and $K(P_{ABE}) > 0$ exist.
\end{theorem}

\begin{proof}[Proof sketch.]
To guarantee $K_{\rm PD}(P_{ABE}) = 0$, we choose an unambiguous $P_{ABE}$ corresponding to a PT-invariant $\rho_{AB}$. Then $D(\rho_{AB}) = 0$ and $K_{\rm PD}$ vanishes by Theorem~\ref{thm:Main}.  We obtain a positive value of $K(P_{ABE})$ by cleverly choosing the diagram of $P_{ABE}$ (see Fig.~\ref{fig:Olives}) and numerically optimizing the right-hand side of
\begin{equation}
  K(P_{ABE}) \geq I(X;B) - I(X;E),
\end{equation}
where $I(X;B)$ denotes the mutual information\fnote{The mutual information between classical random variables $X$ and $B$ is defined as $I(X;B) := H(X) + H(B) - H(XB)$, where $H$ is the entropy function given by \unexpanded{$H(A) := - \sum_a p(a) \log p(a)$}.} between classical random variables $X$ and $B$, and $X$ is obtained by noisy processing of $A$.  Table~\ref{tab:Summary} summarizes our findings for various small dimensions, and Fig.~\ref{fig:Olives} shows the structure of our smallest example, a $3 \times 3$ state.  More details and explicit examples are provided in Appendix~\refx{apx:Examples}.
\end{proof}

\begin{table}[ht]
\begin{center}
\begin{tabular}{|c|r|c|}
  \hline
  $d_A \times d_B$ & $d_E$ & Bits of private key \\
  \hline
  $3 \times 3$ &  $4$ & \csuse{3x3} \\
  $4 \times 4$ &  $6$ & \csuse{4x4} \\
  $4 \times 5$ &  $8$ & \csuse{4x5} \\
  $5 \times 6$ & $10$ & \csuse{5x6} \\
  $6 \times 5$ & $10$ & \csuse{6x5} \\
  \hline
\end{tabular}
\caption{Summary of private bound-entangled states obtained using our construction. Here $d_A$, $d_B$, and $d_E$ are the dimensions of Alice, Bob and Eve. The third column is a numerical lower bound on the amount of distillable private key. The amount of private key in our $4 \times 4$ example exceeds \csuse{4H4} achieved by~\cite{SmallBoundKey}. Our $4 \times 5$ example can be embedded in the $5 \times 6$ and $6 \times 5$ examples, but we report only states that are not trivially reducible to examples in smaller dimensions. This is why the last two examples have smaller key rates despite having larger dimensions.}
\label{tab:Summary}
\end{center}
\end{table}

Since our construction guarantees $D(\psi_{ABE}) = 0$, we can lift any $P_{ABE}$ from Theorem~\ref{thm:cPBE} to a private bound-entangled state $\ket{\psi}_{ABE}$ by applying Theorem~\ref{thm:Main}.

\begin{corollary}\label{cor:PBE}
Any $P_{ABE}$ from Theorem~\ref{thm:cPBE} yields a private bound-entangled state $\ket{\psi}_{ABE}$ via Eq.~(\ref{eq:psiABE}).
\end{corollary}

This gives a new construction of private bound-entangled states (the only known construction before our work was~\cite{BoundKey, SmallBoundKey}).  In fact, due to the lifting established by Theorem~\ref{thm:Main}, it is natural to consider the distribution $P_{ABE}$ in Theorem~\ref{thm:cPBE} as a classical analogue of private bound entanglement.  This provides a satisfactory resolution to the problem of finding a classical analogue of bound entanglement~\cite{GRW02, RW03, ACM04, PA12}.

\section{Implications for classical key agreement} \label{sect:Implications}

The basic technique for classical key agreement is a combination of
Error Correction and Privacy Amplification (\ECPA{}), which achieves a
rate of the mutual information difference $I(A;B)-I(A;E)$~\cite{Wyner75}.  
Essentially all other protocols use \ECPA{}
as a final step.  For example, preceding \ECPA{} by a noisy processing
step in which the distribution of $A$ is modified gives the optimal
key rate for distillation with one-way discussion from Alice to
Bob~\cite{CK78}.  Similarly, Maurer considered public discussion
protocols where Alice and Bob exchange the information about their
variables in a two-way fashion~\cite{Maurer}. Public discussion
includes as special cases post-selection and reverse reconciliation,
but does not include noisy processing.  Maurer showed that two-way
public discussion can be strictly stronger than one-way.  He also
suggested that in the two-way setting noisy processing might give no
benefit~\cite{Maurer}. Evidence suggesting the opposite later was
given in~\cite{AGM06}.

By considering the classical unambiguous probability distributions that yield private bound-entangled states, we find that in general public discussion alone is insufficient for optimal key extraction even in the two-way setting.  Stronger still, while no key can be distilled using only public discussion, a positive rate is achieved by noisy processing and one-way discussion.

\section{Conclusions}

We have concentrated on the analogy between quantum entanglement
distillation and classical key distillation using only public
discussion, and abandoned for now the search for bound information,
which remains an important open question.  This led us to observe the
dual nature of unambiguous distributions and quantum states, which in turn suggested a proof
that noisy processing is necessary for two-way key distillation.
While this finding concerns a purely classical question, reaching this
conclusion appears to require a detour through quantum mechanics---we
know of no classical proof. This suggests an exciting possibility of
using quantum means to solve other questions in classical cryptography
and information theory.

Along the way we found a new construction of private bound-entangled states.  The standard construction involves two systems for each party: a ``key'' system yielding private correlations upon measurement, and a ``shield'' system that weakens Eve's correlation with the key~\cite{BoundKey, SmallBoundKey}.  Our construction does not employ the key/shield distinction.  Instead, we first construct a classical unambiguous probability distribution and promote it to a private bound-entangled quantum state.  This gives an example in $3 \times 3$ dimensions, which is too small to accommodate key and shield subsystems.  We also find an example in $4 \times 4$ with more key than that of~\cite{SmallBoundKey}, and further examples in other dimensions.  Of course, though our constructions do not have a clear key/shield separation, a protocol that distills key from many copies of our states produces trash that cannot be safely handed over to Eve (the state is bound-entangled after all).  This trash can then be identified as the shield of the purified key.

Bound-entangled states are not just a curious mathematical construction---their existence has been verified experimentally~\cite{Mysterious, AB09, LKPR10, BSGM10, KBPS10, DSHP11, DKDB11, DKDB13, ASB13}.  The Smolin state was prepared using polarized photons~\cite{AB09, LKPR10, DKDB13, ASB13} and trapped ions~\cite{BSGM10}.  A pseudo-bound-entangled state was created using nuclear magnetic resonance~\cite{KBPS10}.  A continuous-variable bound-entangled state of light was prepared by~\cite{DSHP11}.  Finally, states with more distillable key than entanglement have been prepared~\cite{DKDB11, DKDB13}, however they are not bound.

So far no experiment has demonstrated a \emph{private} bound-entangled state.  The simplest known such state is given by our construction (Fig.~\ref{fig:Olives}).  It can be prepared by randomly sampling four pure entangled two-qutrit states (three have Schmidt-rank 2 and one has Schmidt-rank 3). Furthermore, their amplitudes are real, so each individual state can be prepared by performing rotations around a single axis in the two-dimensional subspace spanned by $\ket{00}_{AB}$ and $\ket{11}_{AB}$, and permuting the standard basis vectors $\ket{0}$, $\ket{1}$, $\ket{2}$ of each qutrit.

Our work may facilitate an experimental demonstration of superactivation---a phenomenon wherein pairs of quantum channels, neither of which can transmit quantum information on its own, nevertheless have positive capacity when used together~\cite{Super}.  Channels with zero quantum capacity but positive private classical capacity are central to the phenomenon, and these can easily be constructed from our private bound-entangled states.  Indeed, our $3 \times 3$ state gives rise to a zero-capacity channel acting on a single qutrit that can be superactivated by a $50\%$ erasure channel with $4$-dimensional input, the smallest known example.

\section*{Acknowledgments}

We acknowledge Charles Bennett and Debbie Leung for commenting on an earlier version of this manuscript.  We also thank the anonymous referees for bringing reference~\cite{CEHHOR07} to our attention and for other useful suggestions.  This work was supported by DARPA QUEST program under contract number HR0011-09-C-0047.

\bibliographystyle{unsrturl}
\bibliography{PRL}


\onecolumngrid
\vspace{0.2in}


\appendix
\numberwithin{equation}{section}
\numberwithin{figure}{section}

\section{Distillation with remanent devices} \label{apx:Remanence}

Our main contribution is to clarify the distinction between quantum entanglement and key distillation, and to exhibit similar  classical phenomena.  We believe that this distinction is best explained using a model of \emph{remanent devices}, since it allows to use the same language quantumly as well as classically.

Consider the following problem. Let $\ket{\psi}_{ABE}$ be a tripartite state shared between Alice, Bob, and Eve. The goal of Alice and Bob is to use public classical communication to distill a key from $\ket{\psi}_{ABE}$ that is secure from Eve. Moreover, assume the distillation is performed on devices that are susceptible to \emph{data remanence}---that is, when Alice and Bob are done, they take their keys with them, but all other information left on devices (erased or not) becomes available to Eve.

To analyze such distillation protocols, we assume (without loss of generality) that all measurements are deferred till the end of the protocol and at each step a unitary isometry is applied (see Appendix~\ref{apx:Noise} for a classical equivalent of this claim). Let $\ket{\Psi}_{A B T_A T_B E}$ denote the state after the last isometry (see Fig.~\ref{fig:Trash}), where $T_A$ and $T_B$ are the ``trash systems'' of Alice and Bob that are discarded in the final step.

\begin{figure}[!h]


\begin{tikzpicture}

\def\w{1.2cm}; 
\def\h{0.8cm}; 

\path (-0.3*\w,0) coordinate (P);
\node [anchor = east, align = center] at (P) {$P_{ABE}$ \\ or \\ $\ket{\psi}_{ABE}$};

\foreach \n/\i in {B/-1, E/0, A/1} {
  \path (     0,\i*\h) coordinate [label=30:$\n$] (\n);
  \path (3.5*\w,\i*\h) coordinate [label=90:$\n$] (\n');
  \draw (P) -- (\n) -- (\n');
}

\def\rx{0.20}; 
\def\ry{0.05}; 
\def\th{0.25}; 

\foreach \n in {A, B} {
  \path ( \n)+(1.3*\w,-0.2*\h) coordinate (R\n);
  \path (R\n)+(1.2*\w,0) coordinate (T\n);
  \path (T\n)+(0,-0.2) coordinate (t\n);
  \draw (R\n) -- (T\n) -- (t\n);
  \draw (t\n) circle [x radius = \rx, y radius = \ry];
  \path [draw = black] (t\n) ++(\rx,0) -- ++(-0.3*\rx,-\th) .. controls +(220:0.1) and +(-40:0.1) .. ++(-1.4*\rx,0) -- ++(-0.3*\rx,\th);
  \path (t\n) + (3*\rx,0) node {$T_\n$};
}

\path (E)+(1.3*\w,-0.1*\h) node [draw = black, align = center, fill = white, text width = 30, minimum height = 2.6*\h] {LOPC \\ or \\ LOCC};

\draw [densely dashed] (2.1*\w,-1.6*\h) -- (2.1*\w,1.4*\h);
\node [text width = 70, align = center] at (2.1*\w,2.2*\h) {$P_{A B T_A T_B E}$ \\ or \\ $\ket{\Psi}_{A B T_A T_B E}$};

\end{tikzpicture}

\caption{\label{fig:Trash}Any classical or quantum distillation protocol can be cast in the above form. The initial resource is a tripartite probability distribution $P_{ABE}$ or a quantum state $\ket{\psi}_{ABE}$. Alice and Bob perform a sequence of local operations and public communication (LOPC) in the classical case or local operations and classical communication (LOCC) in the quantum case. At each step, either Alice or Bob generates a classical message register that is attached to their system and the systems of the other two parties (for more details see Fig.~\ref{fig:Distillation} in Appendix~\ref{apx:Noise}). Registers $T_A$ and $T_B$ are the trash systems of Alice and Bob which (depending on the setting) may be accessible to Eve.}
\end{figure}
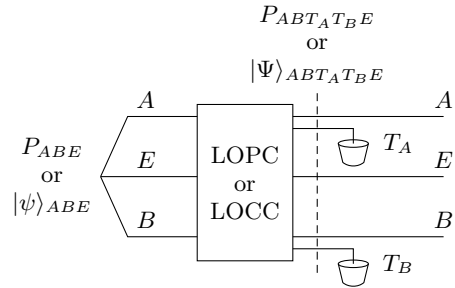

To illustrate the distinction between entanglement and secret key, it is instructive to consider the following two examples, where for simplicity we assume that Eve is decoupled from the other parties and we denote her reduced state by $\ket{\phi}_E$. 

\begin{enumerate}
\item Assume that
\begin{equation}
 \ket{\Psi}_{A B T_A T_B E} =
  \frac{1}{\sqrt{2}} \Bigl(
    \ket{00}_{AB} \ket{\psi_0}_{T_A T_B} +
    \ket{11}_{AB} \ket{\psi_1}_{T_A T_B}
  \Bigr) \ket{\phi}_E,
\end{equation}
where $\ket{\psi_0}_{T_A T_B} \perp \ket{\psi_1}_{T_A T_B}$ are arbitrary states on the trash systems $T_A$ and $T_B$. Alice and Bob might attempt to obtain a shared private bit by discarding $T_A$ and $T_B$. However, its privacy would depend on the assumption that Eve has no access to the discarded systems. If this cannot be guaranteed, the bit is compromised, as Eve can recover it by performing a measurement that perfectly discriminates the orthogonal states $\ket{\psi_0}_{T_A T_B}$ and $\ket{\psi_1}_{T_A T_B}$.
\item On the other hand, if
\begin{equation}
  \ket{\Psi}_{A B T_A T_B E} =
  \frac{1}{\sqrt{2}} \Bigl(
    \ket{00}_{AB} +
    \ket{11}_{AB}
  \Bigr) \ket{\psi}_{T_A T_B} \ket{\phi}_E
\end{equation}
for some arbitrary state $\ket{\psi}_{T_A T_B}$, Eve can learn nothing about the reduced state on $AB$ even when she possesses the trash systems $T_A$ and $T_B$.
\end{enumerate}

In the first case, the state contains a private key as long as Alice and Bob can keep their trash systems $T_A$ and $T_B$ private.  However, if at the end of the protocol Eve can access the remanent data on their devices, she can easily recover the key.  In the second case, the key remains secure even if Eve can access the remanent data.  Note that when the goal is to distill entanglement, allowing Alice and Bob to keep their trash systems makes no difference.

The above way of explaining the distinction between entanglement and key distillation translates in a straightforward way to the classical case.  Appendix~\ref{apx:Noise} explains a classical equivalent of deferring measurements till the end of the protocol.

\section{PT-invariance} \label{apx:PT-invariance}

In this appendix, we describe the PT-invariance condition of $\rho_{AB}$ in terms of the underlying unambiguous distribution $P_{ABE}$. Recall from Eqs.~(\refx{eq:OliveA}) to~(\refx{eq:OliveE}) that $P_{ABE}$ is unambiguous if any two parties can together recover the value of the third party's variable. For example, if Alice has $a$ and Bob has $b$, then Eve's
value is
\begin{equation}
  e(a,b) :=
  \begin{cases}
    e     & \text{if $p(a,b,e) \neq 0$}, \\
    \yuzz & \text{otherwise},
  \end{cases}
  \label{eq:e(a,b)}
\end{equation}
where $\yuzz$ (\emph{yuzz}) is a special symbol that lies outside of Eve's alphabet~\cite{DrSeuss} and indicates that Alice and Bob never have the pair $(a,b)$. Notice that the reduced distribution on Alice and Bob is given by
\begin{equation}
  p(a,b) := \sum_e p(a,b,e) =
  \begin{cases}
    0             & \text{if $e(a,b) = \yuzz$}, \\
    p(a,b,e(a,b)) & \text{otherwise}.
  \end{cases}
  \label{eq:p(a,b)}
\end{equation}

Recall from Eq.~(\refx{eq:PT}) that a bipartite state $\rho_{AB}$ is PT-invariant if $\rho_{AB}\pt = \rho_{AB}$, where the partial transposition is defined on the standard basis as
\begin{equation}
  \bigl( \ketbra{a}{a'}_A \otimes \ketbra{b}{b'}_B \bigr)\pt
  := \ketbra{a}{a'}_A \otimes \ketbra{b'}{b}_B
  \label{eq:PTx}
\end{equation}
and extended by linearity. The following lemma relates PT-invariance of $\rho_{AB}$ to two properties of the underlying unambiguous distribution $P_{ABE}$. The first property says that the diagram associated to $P_{ABE}$ can be obtained by superimposing several crosses (see Fig.~\refx{fig:Olives}), and the second property says that each $2 \times 2$ submatrix corresponding to a cross has rank one.  For example, if $P_{ABE}$ has the diagram shown in Fig.~\refx{fig:Olives}, then the entries of $P_{AB}$ must satisfy
\begin{equation*}
  \det \mx{p_{00} & p_{01} \\ p_{10} & p_{11}} =
  \det \mx{p_{11} & p_{12} \\ p_{21} & p_{22}} =
  \det \mx{p_{00} & p_{02} \\ p_{20} & p_{22}} = 0.
\end{equation*}

\begin{lemma}\label{Lem:PT}
Let $P_{ABE}$ be an unambiguous\footnote{The result is in fact slightly more general, since we only\\need to assume Eq.~(\refx{eq:OliveE}) for the proof to go through.\\} probability distribution. Then the following condition on $P_{ABE}$ is equivalent to $\rho_{AB}$ being PT-invariant: if $e(a,b) = e(a',b') \neq \yuzz$ for some $a \neq a'$ and $b \neq b'$ then
\begin{enumerate}
  \item $e(a,b') = e(a',b) \neq \yuzz$ and
  \item $p(a,b) p(a',b') = p(a,b') p(a',b)$,
\end{enumerate}
where $e(a,b)$ and $p(a,b)$ are defined in Eqs.~(\ref{eq:e(a,b)}) and~(\ref{eq:p(a,b)}), respectively.
\end{lemma}

\begin{proof}
We expand $\rho_{AB}$ using Eqs.~(\refx{eq:rhoAB}) and~(\refx{eq:psiABE}) and compute the partial transpose according to Eq.~\eqref{eq:PTx}:
\begin{align}
   \bigl( \rho_{AB} \bigr)\pt
&= \bigl( \tr_E \proj{\psi}_{ABE} \bigr)\pt \\
&= \sum_e \Biggl( \sum_{a,a',b,b'} \sqrt{p(a,b,e) p(a',b',e)}
   \ketbra{a}{a'}_A \otimes \ketbra{b}{b'}_B \Biggr)\pt \\
&= \sum_e \Biggl( \sum_{a,a',b,b'} \sqrt{p(a,b',e) p(a',b,e)}
   \ketbra{a}{a'}_A \otimes \ketbra{b}{b'}_B \Biggr),
\end{align}
where we relabeled $b$ and $b'$. This is equal to $\rho_{AB}$ if and only if
\begin{equation}
  \forall a,a',b,b':
  \sum_e \sqrt{p(a,b ,e) p(a',b',e)} =
  \sum_e \sqrt{p(a,b',e) p(a',b ,e)}.
\end{equation}
Since $P_{ABE}$ is unambiguous, Eq.~(\refx{eq:OliveE}) implies that each of the two sums contains at most one nonzero term. Moreover, both sides are nonzero exactly when the first condition holds, and equal exactly when the second condition holds.
\end{proof}

\section{Classical-quantum correspondence} \label{apx:Correspondence}

The main reason for introducing unambiguous probability distributions is Theorem~\ref{thm:Mainx} that establishes a relationship between the rate $K_{\rm PD}(P_{ABE})$ of private key that can be distilled from an unambiguous probability distribution $P_{ABE}$ by public discussion, and the distillable entanglement $D(\psi_{ABE})$ of the quantum version $\ket{\psi}_{ABE}$ of that distribution.  The proof of this theorem follows from several lemmas.

\begin{lemma}\label{Lemma:Unambiguity}
Let $P_{ABE}$ be an unambiguous distribution, and suppose $P_{AM,BM,EM}$ can be generated from $P_{ABE}$ by public discussion where $M$ is the public message.  Then the probability distribution $P_{AM,BM,EM}$ is also unambiguous.\footnote{More generally, if $P_{ABE}$ satisfies only some subset of Eqs.~(\refx{eq:OliveA}--\refx{eq:OliveE}), then the same equations are satisfied also by $P_{AM,BM,EM}$.\\}
\end{lemma}

\begin{proof}
It suffices to consider only 1-round protocols, since the general case follows by induction.  Without loss of generality, let the protocol consist of a message $m$ sent from $A$ to $B$ according to some conditional distribution $q(m|a)$.  The probability distribution $P_{AM,BM,EM}$ is then given by $p[(a,m),(b,m),(e,m)] = p(a,b,e) q(m|a)$.  To check that $P_{AM,BM,EM}$ is unambiguous, we fix $(b,m)$ and $(e,m)$ (or equivalently $b,e,m$) and verify that
\begin{align}
      \abs{\set{(a,m) : p[(a,m),(b,m),(e,m)] \neq 0}}
&=    \abs{\set{(a,m) : p(a,b,e) q(m|a) \neq 0}} \\
&\leq \abs{\set{a : p(a,b,e) \neq 0}} \\
&\leq 1,
\end{align}
which is the first condition in Eq.~(\refx{eq:OliveA}).  Similarly, we find the second two conditions are satisfied.
\end{proof}

\begin{lemma}\label{Lemma:LOCC-PD}
If $P_{AM,BM,EM}$ can be generated by public discussion from $P_{ABE}$, then the corresponding quantum state $\rho_{AM,BM}$ can be generated from $\rho_{AB}$ by LOCC.
\end{lemma}

\begin{proof}
We begin by proving the result for one-way protocols from Alice to Bob. Let Alice's message $m$ be chosen according to conditional distribution $q(m|a)$. Then the probabilities of $P_{AM,BM,EM}$ are given in terms of $P_{ABE}$ by $p[(a,m),(b,m),(e,m)] = p(a,b,e) q(m|a)$. In the quantum case, $\rho_{AM,BM}$ can be obtained from $\rho_{AB}$ by having Alice perform a POVM with Kraus operators
\begin{equation}
  A_m = \sum_{a} \sqrt{q(m|a)} \proj{a}_A
\end{equation}
and keeping a copy of $m$ as well as sending it to both Bob and Eve.  The multi-round result follows by repeatedly applying this observation.
\end{proof}

Let us state some definitions that are necessary for the next lemma.  The coherent information of a state $\rho_{AB} = \tr_E \proj{\psi}_{ABE}$ is given by $I(A\rangle B)_{\rho_{AB}} := S(B)-S(E)$.  The advantage of a tripartite distribution $P_{ABE}$ is $A(P_{ABE}) := I(A;B) - I(A;E)$.  It is a lower bound for one-way distillation rate from Alice to Bob.

\begin{lemma}\label{Lemma:CoherentInformation}
Let $P_{ABE}$ be an unambiguous classical tripartite distribution.  Then, 
\begin{equation}
  I(A\rangle B)_{\rho_{AB}} = A(P_{ABE}).
\end{equation}
\end{lemma}

\begin{proof}
The advantage of $P_{ABE}$ is given by 
\begin{align}
  A(P_{ABE}) &= I(A;B) - I(A;E) \\
              &= H(B)-H(B|A) - H(E) + H(E|A).
\end{align}
Since $P_{ABE}$ is unambiguous, from Eqs.~(\refx{eq:OliveB}) and~(\refx{eq:OliveE}), we know that for fixed $a$, conditional distributions 
on $B$ and $E$ are identical 
up to relabeling of the outputs.  Thus, for a fixed $a$ we have $H(B|A=a) = H(E|A=a)$, 
which implies $H(B|A) = H(E|A)$.  This gives us   
\begin{equation}
  A(P_{ABE}) = H(B) - H(E), \label{eq:PrivateInf}
\end{equation}
where the RHS is evaluated on the classical variables $B$ and $E$ distributed according to $p(a,b,e)$.  

Now, note that
\begin{align}\label{eq:CohInf}
  I(A\rangle B)_{\rho_{AB}} = S(B)-S(E),
\end{align}
with the entropies evaluated on $\proj{\psi}_{ABE}$.
From Eqs.~(\refx{eq:OliveE}) and (\refx{eq:OliveB}) applied to $\ket{\psi}_{ABE}$, we find that
\begin{align}
  \rho_{E}  & = \sum_{a,b,e}p(a,b,e)\proj{e}, \\ 
  \rho_{B}  & = \sum_{a,b,e}p(a,b,e)\proj{b}, 
\end{align}
so that the von~Neumann entropies on the right-hand side of Eq.~(\ref{eq:CohInf}) are identical to the Shannon entropies in Eq.~(\ref{eq:PrivateInf}),
which proves the result.
\end{proof}

Notice that the above proof used only Eqs.~(\refx{eq:OliveE}) and (\refx{eq:OliveB}).  If the roles of systems $A$ and $B$ are exchanged, one needs Eqs.~(\refx{eq:OliveE}) and (\refx{eq:OliveA}) instead.  In the following theorem, the roles of $A$ and $B$ are not known in advance, so we demand that $P_{ABE}$ is unambiguous (\emph{i.e.}, satisfies all three equations).  The first half of our proof relies on the above lemma with the roles of $A$ and $B$ possibly exchanged.

\begin{theorem}\label{thm:Mainx}
Let $P_{ABE}$ be an unambiguous probability distribution and $\ket{\psi}_{ABE}$ be the associated quantum state. The distillable entanglement of $\ket{\psi}_{ABE}$ is at least as big as the distillable key by public discussion of $P_{ABE}$,
\begin{equation}
  D(\psi_{ABE}) \geq K_{\rm PD}(P_{ABE}).
  \label{eq:Noiseless}
\end{equation}
The distillable key of $\ket{\psi}_{ABE}$ is at least as big as the distillable key by public discussion and noisy processing of $P_{ABE}$,
\begin{equation}
  K(\psi_{ABE}) \geq K(P_{ABE}).
  \label{eq:Noisy}
\end{equation}
\end{theorem}

\begin{proof}
Suppose we can achieve a key rate $R$ by using public discussion to distill from $P_{ABE}$.  Then, for every $\delta > 0$ there is an $n \geq 1$ and a public discussion protocol with message history $M$ yielding distribution $P_{A^n\!M, B^n\!M, E^n\!M}$ that gives advantage $\frac{1}{n} A(P_{A^n\!M, B^n\!M, E^n\!M}) > R - \delta$, where the roles of $A$ and $B$ might possibly be exchanged.  By Lemma~\ref{Lemma:Unambiguity} we know that $P_{A^n\!M, B^n\!M, E^n\!M}$ is unambiguous, so that Lemma~\ref{Lemma:LOCC-PD} implies that the corresponding quantum state $\rho_{A^n\!M, B^n\!M}$ can be generated from $n$ copies of $\rho_{AB}$ by LOCC.  Furthermore, Lemma~\ref{Lemma:CoherentInformation} implies that the coherent information of $\rho_{A^n\!M, B^n\!M}$ is equal to $A(P_{A^n\!M, B^n\!M, E^n\!M})$, where the roles of $A$ and $B$ again might be exchanged.  Since the coherent information is an achievable rate of entanglement distillation, we thus find that the distillable entanglement of $\rho_{AB}$ is at least and can be made arbitrarily close to $R$.  This establishes the first part of the theorem.

For the second part, Alice and Bob use a particular ``classical'' strategy to distill the key from $\ket{\psi}_{ABE}$. That is, they make local copies of the variables they have and then proceed with the classical protocol. The local copies ensure that Eve is dephased in the standard basis. The security of the classical protocol implies that Eve is ignorant of the key also in the quantum case.
\end{proof}

The second part of this theorem has already been obtained by Christandl, Ekert, Horodecki, Horodecki, Oppenheim, and Renner in a slightly different form (see Corollary~3 in their paper).  They only assume that $P_{ABE}$ satisfies Eq.~(\refx{eq:OliveE}), which is indeed sufficient to derive Eq.~(\ref{eq:Noisy}).  However, in addition they also claim equality in Eq.~(\ref{eq:Noisy}) but provide no proof.  We also believe that equality might hold and leave it as an open problem.  A related open problem asks whether equality holds in Eq.~(\ref{eq:Noiseless}).

\section{Deferral of noisy processing} \label{apx:Noise}

In this appendix we argue that classical randomized private key distillation protocols can without loss of generality be cast in a specific form.  In the main text we described two types of protocols (see Fig.~\ref{fig:Distillation}): ones that involve a noisy processing step, which can modify the local random variables by a stochastic map, and ones that do not. The following lemma shows that these two types of protocols are the most general ones for the cases of not having noisy processing and having noisy processing, respectively.

\begin{figure}[!ht]



\definecolor{darkorange}{rgb}{0.9, 0.4, 0.0}
\definecolor{lightgreen}{rgb}{0.7, 0.9, 0.7}

\begin{tikzpicture}[
  > = latex',
  puffy/.style = {cloud, cloud puffs = 11, cloud ignores aspect, fill = orange!40, draw = black},
  noise/.style = {starburst, starburst point height = 5pt, draw = darkorange, fill = yellow},
  box/.style = {draw = black, fill = lightgreen},
  reg/.style = {
}]

\def\w{1.9}; 
\def\h{1.1}; 

\node (P) at (0,0) [puffy] {$P_{ABE}$};
\foreach \n/\i in {A/-1, E/0, B/1} {
  \node (\n) at (\i*\w,-\h) {$\n$};
  \draw [->] (P) to (\n);
}
\path (A) + (0,-1.0*\h) node[reg] (A1) {$A M_1$};
\path (B) + (0,-1.5*\h) node[reg] (B1) {$B M_1$};
\draw [->] (A) to (A1);
\draw [->] (B) to (B1);
\draw [->] (A1) to node[above] {$M_1$} (B1);
\path (B1) + (0,-1.0*\h) node[reg] (B2) {$B M_1 M_2$};
\path (A1) + (0,-2.0*\h) node[reg] (A2) {$A M_1 M_2$};
\draw [->] (A1) to (A2);
\draw [->] (B1) to (B2);
\draw [->] (B2) to node[above] {$M_2$} (A2);
\path (B2) + (0,-0.60*\h) node {\vdots};
\path (A2) + (0,-0.60*\h) node {\vdots};
\path (A2) + (\w,-0.35*\h) node {\vdots};
\path (A2) + (0,-1.5*\h) node[reg] (An) {$A M_1 \dots M_n$};
\path (B2) + (0,-1.5*\h) node[reg] (Bn) {$B M_1 \dots M_n$};
\draw [->] (Bn) to node[above] {$M_n$} (An);
\path (An) + (0,-1.1*\h) node[noise] (A') {$\tilde{A}$};
\path (Bn) + (0,-1.6*\h) node[noise] (B') {$\tilde{B}$};
\path (A') + (\w,+6pt) node {noisy local};
\path (A') + (\w,-6pt) node {processing};
\draw [->] (An) to (A');
\draw [->] (Bn) to (B');
\path (A') + (0,-1.0*\h) node[box] (a) {\ECPA};
\path (B') + (0,-1.5*\h) node[box] (b) {\ECPA};
\draw [->] (A') to (a);
\draw [->] (B') to (b);
\draw [->] (a) to (b);
\path (a) + (0,-1.3*\h) node (KA) {$K$};
\path (b) + (0,-0.8*\h) node (KB) {$K$};
\draw [->] (a) to (KA);
\draw [->] (b) to (KB);

\end{tikzpicture}


\caption{\label{fig:Distillation}
Alice and Bob use public discussion to distill a secret key $K$ from $P_{ABE}$ in the presence of an eavesdropper Eve.  At the $i$th step of the protocol, either Alice or Bob produces a public message $M_i$ that stochastically depends on her/his respective variables $A$ and $B$ and all messages from previous rounds$^9$; it becomes available to the eavesdropper as well as the legitimate recipient.  Each such operation is an isometry$^{10}$, since it only introduces a new register.  This explicitly keeps track of the trash registers accumulated during the protocol.  In the final step, error correction and privacy amplification (\ECPA{}) are applied on the original variables and the messages.  \ECPA{} is a one-way protocol and can be performed either from Alice to Bob or Bob to Alice.  Before \ECPA{}, a local noisy processing step may be included. It corresponds to modifying one's random variables by applying a stochastic map.  Without noisy processing, no key can be extracted from PT-invariant unambiguous distributions.  However, allowing noisy processing at the last step is sufficient to obtain the most general randomized distillation protocol (see Lemma~\ref{lem:Deferral}).  Indeed, such protocol can extract secret key from our distributions listed in Appendix~\ref{apx:Examples}.}
\end{figure}
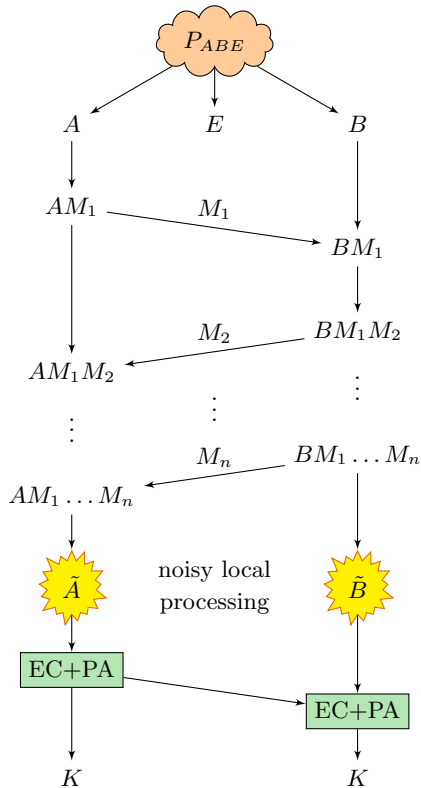

\vspace{0.5cm}

\begin{lemma}\label{lem:Deferral}
The following holds:
\begin{enumerate}
  \item if noisy processing is not involved, there is no advantage for Alice and Bob to introduce extra local random variables other than those initially given to them (\textit{i.e.}, $A$ for Alice and $B$ for Bob);
  \item if noisy processing is allowed throughout the protocol, it can always be deferred till the very last step.
\end{enumerate}
\end{lemma}

\footnotetext{Lemma~\ref{lem:Deferral} shows that this is without loss of generality, \emph{i.e.}, Alice and Bob do not need to introduce new random variables or modify existing ones.}
\footnotetext{An isometry preserves the $2$-norm in the quantum case and the $1$-norm in the classical case. A classical isometry is a stochastic map that sends different standard basis vectors to probability distributions with disjoint supports.}

\begin{proof}
For the first claim, note that each step of any classical randomized protocol can be described in terms of a conditional probability distribution: the probability of generating a particular value of any new random variable in terms of already existing variables.  These can be used to generate a joint probability distribution of all variables---local variables and messages alike---throughout the protocol.  Specifically, at the end of the protocol these give a distribution on $(A, A_1, \dotsc, A_n, B, B_1, \dotsc, B_n, M_1, \dotsc, M_n)$, where $A_i$ and $B_i$ are the local random variables generated in the $i$th round of the protocol by Alice and Bob, respectively, $M_i$ is the $i$th message sent, and $n$ is the number of rounds.  From this, we can compute a set of conditional distributions for messages $M_i$ conditioned solely on the previous messages and either $A$ or $B$ (depending on whether Alice or Bob generates $M_i$).  Thus, in the last step Alice can first generate $(A, M_1, \dotsc, M_n)$ and then use the conditional probability distributions for the $A_i$ given $(A, M_1, \dotsc, M_n)$ to generate the remaining local random variables $(A_1, \dotsc, A_n)$.  Bob can generate his local random variables similarly.  However, generating these extra variables at the last step of the protocol does not affect the final rate achieved by \ECPA{}. This follows by using the chain rule for the mutual information.

For the second claim, let's assume that Alice and Bob use local noisy processing at every step of the protocol, and let $\tilde{A}_i$ and $\tilde{B}_i$ denote their local random variables at step $i$. For example, Alice produces $\tilde{A}_i$ and $M_i$ from $(\tilde{A}_{i-1}, M_1, \dotsc, M_{i-1})$ by a stochastic map.  However, imagine that instead of destroying her previous random variable $\tilde{A}_{i-1}$ every time, Alice makes a local copy of it and keeps it around.  Such protocol has exactly the same form as described in the first scenario, except in the last step both parties have to destroy all their local variables except the last one.  Now we can apply the same argument as before and deffer the creation of the local variables till the last step.  After that each party destroys all their local variables except the last one. This yields an equivalent randomized protocol, where noisy processing is applied only in the last step.
\end{proof}

\section{Construction and examples} \label{apx:Examples}

In this appendix we discuss our construction in more detail and provide several examples.  A summary of our examples is given in a table in the main text.  These examples are obtained as follows.  First, we choose a graph that determines the combinatorial structure of $P_{ABE}$ and guarantees that it is unambiguous (see Fig.~\refx{fig:Olives} in the main text).  Next, we parametrize the reduced distribution $P_{AB}$ and impose the restrictions from Lemma~\ref{Lem:PT}, which guarantee that $\rho_{AB}$ is PT-invariant; we also parametrize the noisy processing map $Q_{X|A}$.  Finally, we optimize the parameters numerically to maximize the amount of distillable private key.

\begin{theorem}\label{thm:cPBEx}
Unambiguous probability distributions $P_{ABE}$ with $K_{\rm PD}(P_{ABE}) = 0$ and $K(P_{ABE}) > 0$ exist.
\end{theorem}

\begin{proof}
First, we fix $d_A$ and $d_B$, the dimensions of $A$ and $B$, and choose a diagram in the $d_A \times d_B$ grid that encodes the combinatorial structure of $P_{ABE}$.  We specify a vertex in this diagram by $(a,b) \in A \times B$ and use $(a,b) - (a',b')$ to denote adjacent vertices.  The diagram must satisfy the following two properties:
\begin{enumerate}
  \item it is a union of disjoint diagonal cliques (a \emph{clique} is a collection of vertices such that each pair is adjacent; a clique is \emph{diagonal} if for every pair $(a,b) \neq (a',b')$ we have $a \neq a'$ and $b \neq b'$);
  \item it is a union of overlapping crosses (a \emph{cross} is a pair of edges $(a,b) - (a',b')$ and $(a,b') - (a',b)$ for some $a \neq a'$ and $b \neq b'$).
\end{enumerate}
The first condition guarantees that the distribution is unambiguous. The second condition is necessary for PT-invariance (see Lemma~\ref{Lem:PT}).  The number of cliques in the diagram determines $d_E$, the dimensions of $E$.

Next, we introduce variables $\set{P_{AB}(a,b) : (a,b) \in A \times B}$ that parametrize the distribution $P_{AB}$.  Recall that $P_{AB}$ together with the diagram determines a tripartite unambiguous distribution $P_{ABE}$.  To describe the noisy processing, we introduce a conditional distribution $Q_{X|A}$ parametrized by $\set{Q_{X|A}(x|a) : (x,a) \in X \times A}$. These variables are subject to normalization constraints
\begin{align}
  \sum_{\substack{a \in A \\ b \in B}} P_{AB}(a,b) &= 1, & \forall a \in A: \sum_{x \in X} Q_{X|A}(x|a) &= 1.
  \label{eq:Normalization}
\end{align}
Furthermore, $P_{AB}$ is also subject to constraint
\begin{equation}
  \det \bigl[ P_{AB}(\set{a,a'} \times \set{b,b'}) \bigr] = 0
  \label{eq:Cross}
\end{equation}
for every $2 \times 2$ submatrix of $P_{AB}$ corresponding to a cross formed by edges $(a,b) - (a',b')$ and $(a,b') - (a',b)$.  Under these constraints, Lemma~\ref{Lem:PT} implies that the associated quantum state $\rho_{AB}$ is PT-invariant.  Hence, $D(\rho_{AB}) = 0$ and we are guaranteed by Theorem~\ref{thm:Mainx} that $K_{\rm PD}(P_{ABE}) = 0$.

To obtain a positive value for $K(P_{ABE})$, we numerically optimize
\begin{equation}
  \max_{P_{AB}, Q_{X|A}} I(X;B) - I(X;E)
  \label{eq:max}
\end{equation}
subject to constraints in Eqs.~(\ref{eq:Normalization}) and~(\ref{eq:Cross}).  Here the mutual informations are evaluated on the distribution $P_{XBE}$ defined via
\begin{equation}
  P_{XBE}(x,b,e) := \sum_{a \in A} Q_{X|A}(x|a) P_{ABE}(a,b,e).
\end{equation}
A table in the main text summarizes our findings for various small dimensions.  The structure diagrams of $P_{ABE}$ together with explicit values of $P_{AB}$ and $Q_{X|A}$ are provided on the last page.
\end{proof}

\begin{figure}[!ht]

  \input{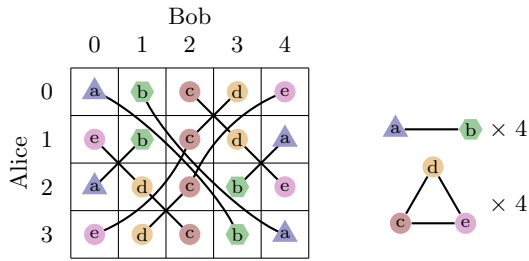}

\caption{Diagram for the $4 \times 5$ example with additional structure imposed on $P_{AB}$. The convention we use to represent unambiguous distributions is explained in Fig.~\refx{fig:Olives} in the main text. In addition, here we use different types of nodes to represent different variables in our parametrization of $P_{AB}$ in Eq.~(\ref{eq:P and Q}). Note that for each $e \in E$, the reduced distribution on $AB$ is perfectly correlated between $A$ and $B$ (up to relabeling of outputs) and is given (up to normalization) by either $(a,b)$ or $(c,d,e)$, each appearing exactly four times.}
\end{figure}

\begin{example}
To illustrate our method, consider the $4 \times 5$ diagram shown above.  In this case, we can choose the parametrization of $P_{AB}$ and $Q_{X|A}$ so that the optimization problem becomes especially simple:
\begin{align}
  P_{AB} &= \frac{1}{4} \mx{
    a & b & c & d & e \\
    e & b & c & d & a \\
    a & d & c & b & e \\
    e & d & c & b & a } & \text{and} &&
  Q_{X|A} &= \mx{
    1 & 1 & 0 & 0 \\
    0 & 0 & 1 & 1 },
  \label{eq:P and Q}
\end{align}
where $a,b,c,d,e \geq 0$ are such that $a + b + c + d + e = 1$.  This automatically satisfies Eq.~(\ref{eq:Normalization}).  Furthermore, we get Eq.~(\ref{eq:Cross}) by simply imposing $a b = d e$, which can be seen by inspecting the crosses in the diagram above.  One can verify that the objective function in Eq.~(\ref{eq:max}) can be expressed as
\begin{equation}
  f(a,b,c,d,e) :=
    - a - c - e + h(a) - h(d) + h(e)
    - h(a + b) + h(b + d) + h(c + d) - h(c + d + e),
\end{equation}
where $h(p) := - p \log_2 p$.  A simple solution can be obtained by choosing
\begin{equation}
  f \left( \frac{1}{10}, \frac{1}{10}, \frac{3}{8}, \frac{1}{40}, \frac{2}{5} \right) \approx 0.0347590.
\end{equation}
The optimal solution can be found numerically and is given on the last page.  It happens to have the same structure as Eq.~(\ref{eq:P and Q}), even though we did not impose any structure on $P_{AB}$ and $Q_{X|A}$ in our program (the only input to our program is a diagram that specifies the location of non-zero entries of $P_{ABE}$).
\end{example}

On the next page we provide a table of unambiguous probability distributions $P_{ABE}$ in various small dimensions obtained using the method described in Theorem~\ref{thm:cPBEx}.  We use the graphical representation explained in Fig.~\refx{fig:Olives} (in the main text) to specify the combinatorial structure of $P_{ABE}$.  For each distribution, we provide the diagram that we chose.  We also provide a numerical lower bound for $K(P_{ABE})$---obtained by numerically optimizing Eq.~(\ref{eq:max})---together with the optimal distribution $P_{AB}$.  Rows of $P_{AB}$ correspond to Alice and columns to Bob, and symbol $\yuzz$ indicates that Alice and Bob's variables never have the corresponding value (see Eq.~(\ref{eq:e(a,b)}) for more details).  In addition, we also list the optimal conditional distribution $Q_{X|A}$ that describes the noisy processing performed by Alice to obtain a random variable $X$ from $A$ (rows of $Q_{X|A}$ correspond to $X$ and columns correspond to Alice). In all cases, we chose $X$ to be of dimension two which was sufficient for obtaining a positive rate.


\csdef{P3x3}{
$P_{AB} =
\begin{pmatrix}
 0.167184 & 0.171529 & 0.001243 \\
 0.089041 & 0.091355 & 0.017492 \\
 0.441714 & 0.017157 & 0.003285
\end{pmatrix}$
}

\csdef{Q3x3}{
$Q_{X|A} =
\begin{pmatrix}
 1 & 0 & 0.670965 \\
 0 & 1 & 0.329035
\end{pmatrix}$
}


\csdef{P4x4}{
$P_{AB} =
\begin{pmatrix}
 0.024798 & 0.119200 & 0.128999 & 0.009393 \\
 0        & 0.087320 & 0.094498 & 0.035793 \\
 0.128999 & 0.119200 & 0.024798 & 0.009393 \\
 0.094498 & 0.087320 & 0        & 0.035793
\end{pmatrix}$
}

\csdef{Q4x4}{
$Q_{X|A} =
\begin{pmatrix}
 1 & 0 & 1 & 0 \\
 0 & 1 & 0 & 1
\end{pmatrix}$
}


\csdef{P4x5}{
$P_{AB} =
\begin{pmatrix}
 0.015228 & 0.033970 & 0.092123 & 0.004989 & 0.103690 \\
 0.103690 & 0.033970 & 0.092123 & 0.004989 & 0.015228 \\
 0.015228 & 0.004989 & 0.092123 & 0.033970 & 0.103690 \\
 0.103690 & 0.004989 & 0.092123 & 0.033970 & 0.015228
\end{pmatrix}$
}

\csdef{Q4x5}{
$Q_{X|A} =
\begin{pmatrix}
 1 & 1 & 0 & 0 \\
 0 & 0 & 1 & 1
\end{pmatrix}$
}


\csdef{P5x6}{
$P_{AB} =
\begin{pmatrix}
 0.076349 & 0.004299 & 0.070542 & 0        & 0        & 0.014384 \\
 0.014674 & 0.006016 & 0.098724 & 0.006016 & 0.014674 & 0.098724 \\
 0.050896 & 0.020867 & 0.047025 & 0.020867 & 0.050896 & 0.047025 \\
 0        & 0        & 0.014384 & 0.004299 & 0.076349 & 0.070542 \\
 0        & 0.022142 & 0.074083 & 0.022142 & 0        & 0.074083
\end{pmatrix}$
}

\csdef{Q5x6}{
$Q_{X|A} =
\begin{pmatrix}
 0 & 0 & 1 & 0 & 1 \\
 1 & 1 & 0 & 1 & 0
\end{pmatrix}$
}


\csdef{P6x5}{
$P_{AB} =
\begin{pmatrix}
 0.026574 & 0.061138 & 0.065969 & 0        & 0        \\
 0.003409 & 0.056660 & 0.061138 & 0        & 0.011779 \\
 0.004843 & 0.080489 & 0.012023 & 0.026034 & 0.089945 \\
 0        & 0.056660 & 0.061138 & 0.003409 & 0.011779 \\
 0        & 0.061138 & 0.065969 & 0.026574 & 0        \\
 0.026034 & 0.080489 & 0.012023 & 0.004843 & 0.089945
\end{pmatrix}$
}

\csdef{Q6x5}{
$Q_{X|A} =
\begin{pmatrix}
 0 & 1 & 1 & 0 & 1 & 0 \\
 1 & 0 & 0 & 1 & 0 & 1
\end{pmatrix}$
}


\newpage
\newgeometry{vmargin = 0.9in, hmargin = 0.58in}

\begin{tabular}{|>{\centering}m{5.1cm}|@{\hspace{10pt}}m{11.6cm}|}
  \hline
  
  \drawAB{3}{3}{
  \draw[state] (00) -- (11) -- (22);
  \draw[state] (01) -- (10);
  \draw[state] (12) -- (21);
  \draw[state] (02) .. controls +(225-\dangle:\step) and +(45+\dangle:\step).. (20);
  \path (21) + (0,-1.2*\step) node {$K(P_{ABE}) \geq \csuse{3x3}$};
}

 & \csuse{P3x3} \newline \vspace{8pt} \newline \csuse{Q3x3} \\ \hline
  
  \drawAB{4}{4}{
  \emptynode{10};
  \emptynode{32};
  \draw[state] (00) .. controls +(-45-\dangle:\step) and +(135+\dangle:\step).. (33);
  \draw[state] (20) -- (31);
  \draw[state] (01) -- (12) -- (23);
  \draw[state] (02) -- (11);
  \draw[state] (13) -- (22);
  \draw[state] (03) .. controls +(225-\dangle:\step) and +(45+\dangle:\step).. (21) -- (30);
  \path (31) + (0.5*\step,-1.2*\step) node {$K(P_{ABE}) \geq \csuse{4x4}$};
}
 & \csuse{P4x4} \newline \vspace{8pt} \newline \csuse{Q4x4} \\ \hline
  
  \drawAB{4}{5}{
  \draw[state] (01) .. controls +(-45-\dangle:\step) and +(135+\dangle:\step).. (34);
  \draw[state] (00) .. controls +(-45+\dangle:\step) and +(135-\dangle:\step).. (33);
  \draw[state] (10) -- (21) -- (32);
  \draw[state] (02) -- (13) -- (24);
  \draw[state] (03) -- (12) .. controls +(225+\dangle:\step) and +(45-\dangle:\step).. (30);
  \draw[state] (14) -- (23);
  \draw[state] (20) -- (11);
  \draw[state] (04) .. controls +(225-\dangle:\step) and +(45+\dangle:\step).. (22) -- (31);
  \path (32) + (0,-1.2*\step) node {$K(P_{ABE}) \geq \csuse{4x5}$};
}
 & \csuse{P4x5} \newline \vspace{8pt} \newline \csuse{Q4x5} \\ \hline
  
  \drawAB{5}{6}{
  \emptynode{03};
  \emptynode{04};
  \emptynode{30};
  \emptynode{31};
  \emptynode{40};
  \emptynode{44};
  \draw[state] (01) -- (12) .. controls +(-45-\dangle:\step) and +(135+\dangle:\step).. (45);
  \draw[state] (10) -- (21);
  \draw[state] (20) -- (11) -- (02);
  \draw[state] (00) .. controls +(-45+\dangle:\step) and +(135-\dangle:\step).. (22);
  \draw[state] (41) .. controls +( 45+\dangle:\step) and +(225-\dangle:\step).. (05);
  \draw[state] (32) -- (43);
  \draw[state] (13) -- (24) -- (35);
  \draw[state] (14) -- (23);
  \draw[state] (25) -- (34);
  \draw[state] (15) .. controls +(225-\dangle:\step) and +( 45+\dangle:\step).. (33) -- (42);
  \path (42) + (0.5*\step,-1.2*\step) node {$K(P_{ABE}) \geq \csuse{5x6}$};
}
 & \csuse{P5x6} \newline \vspace{8pt} \newline \csuse{Q5x6} \\ \hline
  
  \drawAB{6}{5}{
  \emptynode{30};
  \emptynode{40};
  \emptynode{03};
  \emptynode{13};
  \emptynode{04};
  \emptynode{44};
  \draw[state] (10) -- (21) .. controls +(-45+\dangle:\step) and +(135-\dangle:\step).. (54);
  \draw[state] (01) -- (12);
  \draw[state] (02) -- (11) -- (20);
  \draw[state] (00) .. controls +(-45+\dangle:\step) and +(135-\dangle:\step).. (22);
  \draw[state] (14) .. controls +(225-\dangle:\step) and +( 45+\dangle:\step).. (50);
  \draw[state] (23) -- (34);
  \draw[state] (31) -- (42) -- (53);
  \draw[state] (41) -- (32);
  \draw[state] (52) -- (43);
  \draw[state] (51) .. controls +( 45+\dangle:\step) and +(225-\dangle:\step).. (33) -- (24);
  \path (52) + (0,-1.2*\step) node {$K(P_{ABE}) \geq \csuse{6x5}$};
}
 & \csuse{P6x5} \newline \vspace{8pt} \newline \csuse{Q6x5} \\ \hline
\end{tabular}

\end{document}